\documentclass[%
reprint,
superscriptaddress,
showpacs,
amsmath,amssymb,
aps,
pra,
]{revtex4-1}
\usepackage[normalem]{ulem}
\usepackage{graphicx}
\usepackage{color}
\usepackage{amsthm}
\usepackage{amsmath}	
 \usepackage{mathtools}

\def\dim{\mathop{\mathrm{dim}}}
\def\ker{\mathop{\mathrm{ker}}}

\begin{document}

\newcommand{\BigFig}[1]{\parbox{12pt}{\Huge #1}}
\newcommand{\BigZero}{\BigFig{0}}

\renewcommand{\topfraction}{0.99}
\renewcommand{\bottomfraction}{0.99}
\renewcommand{\dbltopfraction}{0.99}
\renewcommand{\textfraction}{0.1}
\renewcommand{\floatpagefraction}{0.99}
\renewcommand{\dblfloatpagefraction}{0.99}

 \newtheorem{theo}{Theorem}
  \newtheorem{lem}{Lemma}

\def\soutr#1{{\color{red}\sout{#1}}}
\def\red#1{{\color{red}#1}}

\def\U#1{{%
\def\O{\mbox{O}}
\def\u{\mbox{u}}
\mathcode`\u=\mu
\mathcode`\O=\Omega
\mathrm{#1}}}
\def\ii{{\mathrm{i}}}
\def\jj{\,\mathrm{j}}                   
\def\ee{{\mathrm{e}}}
\def\dd{{\mathrm{d}}}
\def\cc{{\mathrm{c.c.}}}
\def\Re{\mathop{\mathrm{Re}}}
\def\Im{\mathop{\mathrm{Im}}}
\def\bra#1{\langle #1|}
\def\braa#1{\langle\langle #1|}
\def\ket#1{|\mbox{$#1$}\rangle}
\def\kett#1{|\mbox{$#1$}\rangle\rangle}
\def\bracket#1{\langle\mbox{$#1$}\rangle}
\def\bracketi#1#2{\langle\mbox{$#1$}|\mbox{$#2$}\rangle}
\def\bracketii#1#2#3{\langle\mbox{$#1$}|\mbox{$#2$}|\mbox{$#3$}\rangle}
\def\vct#1{{\mathchoice{\mbox{\boldmath$#1$}}{\mbox{\boldmath$#1$}}%
  {\mbox{\scriptsize\boldmath$#1$}}{\mbox{\scriptsize\boldmath$#1$}}}}
\def\thrvct#1{\mathbf{#1}}
\def\fracpd#1#2{\frac{\partial#1}{\partial#2}}
\def\rank{\mathop{\mathrm{rank}}} 
\def\sub#1{_{\mbox{\scriptsize \rm #1}}}
\def\sur#1{^{\mbox{\scriptsize \rm #1}}}
\def\kagome{kagom\'{e} }
\def\Kagome{Kagom\'{e} }
\def\kagometype{kagom\'{e}-type }
\def\Kagometype{Kagom\'{e}-type }
\def\deg{\mathop{\mathrm{deg}}} 

\title{
Supersymmetric correspondence in spectra on a graph and its line graph: From circuit theory to spoof plasmons on metallic lattices 
}
\author{Yosuke~Nakata}
\email{y\_nakata@shinshu-u.ac.jp}
\affiliation{Center for Energy and Environmental Science, Shinshu University, 4-17-1 Wakasato, Nagano 380-8553, Japan}
\author{Yoshiro~Urade}
\affiliation{Department of Electronic Science and Engineering, Kyoto University, Kyoto 615-8510, Japan}
\author{Toshihiro~Nakanishi}
\affiliation{Department of Electronic Science and Engineering, Kyoto University, Kyoto 615-8510, Japan}
\author{Fumiaki~Miyamaru}
\affiliation{Center for Energy and Environmental Science, Shinshu University, 4-17-1 Wakasato, Nagano 380-8553, Japan}
\affiliation{Department of Physics, Faculty of Science, Shinshu University, 3-1-1 Asahi, Matsumoto, Nagano 390-8621, Japan}
\author{Mitsuo~Wada~Takeda}
\affiliation{Department of Physics, Faculty of Science, Shinshu University, 3-1-1 Asahi, Matsumoto, Nagano 390-8621, Japan}
\author{Masao~Kitano}
\affiliation{Department of Electronic Science and Engineering, Kyoto University, Kyoto 615-8510, Japan}

\date{\today}
\begin{abstract}
We investigate the supersymmetry (SUSY) structures for inductor-capacitor circuit networks on a simple regular graph and its line graph.
We show that their eigenspectra must coincide (except, possibly, for the highest eigenfrequency) due to SUSY, which is derived from the topological nature of the circuits.
To observe this spectra correspondence in the high frequency range, we study spoof plasmons on metallic hexagonal and \kagome lattices.
The band correspondence between them is predicted by a simulation.
Using terahertz time-domain spectroscopy, we demonstrate the band correspondence of fabricated metallic hexagonal and \kagome lattices.
\end{abstract}

\maketitle

\section{Introduction}

Supersymmetry (SUSY) is a conjectured symmetry between fermions and bosons.
Although the concept of SUSY was introduced in high-energy physics and remains to be experimentally confirmed, the underlying algebra is also found in quantum mechanics.
When the SUSY algebra is applied to the field of quantum mechanics it is called supersymmetric quantum mechanics (SUSYQM) \cite{Cooper1994}.
The algebraic relations of SUSY link two systems that at first glance might seem to be very different.
The linkage through SUSY can be utilized to construct exact solutions for various systems in quantum mechanics.
Recently, SUSYQM has been applied to construct quantum systems enabling exotic quantum wave propagations:
reflectionless or invisible defects in tight-binding models \cite{Longhi2010} 
and complex crystals \cite{Longhi2013a},
transparent interface between two isospectral one-dimensional crystals \cite{Longhi2013},
reflectionless bent waveguides for matter-waves \cite{Campo2014}, 
and disordered systems with Bloch-like eigenstates and band gaps \cite{Yu2015}.

The SUSY structure was also found in other physics fields besides quantum mechanics, e.g., statistical physics through the Fokker-Planck equations \cite{Bernstein1984}.
Through the similarity between quantum-mechanical probability waves and electromagnetic waves, the SUSY structure can be formulated for electromagnetic systems.
Electromagnetic SUSY structures have been found in one-dimensional refractive index distributions \cite{Chumakov1994, Miri2013},
coupled discrete waveguides \cite{Longhi2010, Miri2013},
weakly guiding optical fibers with cylindrical symmetry \cite{Miri2013},
planar waveguides with varying permittivity and permeability \cite{Laba2014},
and non-uniform grating structures \cite{Longhi2015}.
Even a quantum optical deformed oscillator with $\mathrm{SU}(1,1)$ group symmetry and its SUSY partner were constructed as a classical electromagnetic system \cite{Zuniga-Segundo2014}.

The SUSY transformation generates new optical systems whose spectra coincide with those of the original system (except possibly for the highest eigenvalue of the fundamental mode of original or generated systems).
The SUSY transformations have been utilized to synthesize mode filters \cite{Miri2013} and distributed-feedback filters with any desired number of resonances at the target frequencies \cite{Longhi2015}.
The scattering properties of the optical systems paired by the SUSY transformation are related to each other \cite{Longhi2010, Miri2013}.
It is possible to design an optical system family with identical reflection and transmission characteristics by using the SUSY transformations \cite{Miri2014}.
A reflectionless potential derived from the trivial system by SUSY transformation was applied to design transparent optical intersections \cite{Longhi2015a}.
Moreover, SUSY has also been intensively investigated in non-Hermitian optical systems.
If a system is invariant under the simultaneous operations of
the space and time inversions, it is called $\mathcal{PT}$-symmetric.
The SUSY transformation for the $\mathcal{PT}$-symmetric system allows for arbitrarily removing bound states from the spectrum \cite{Miri2013a}.
In addition, non-Hermitian optical couplers can be designed \cite{Principe2015}.
By using double SUSY transformations, the bound states in the continuum were also formulated in tight-binding lattices \cite{Longhi2014,Longhi2014a} and continuous systems \cite{Correa2015}.
The SUSY transformation in the $\mathcal{PT}$-symmetric system can also reduce the undesired reflection of one-way-invisible optical crystals \cite{Midya2014}.

From an experimental perspective, it is still challenging to extract the full potential of electromagnetic SUSY because of fabrication difficulties.
However, using dielectric coupled waveguides, researchers have realized a reflectionless potential \cite{Szameit2011}, interpreted as a transformed potential derived from the trivial one by a SUSY transformation \cite{Longhi2010}, and SUSY mode converters \cite{Heinrich2014}.
The SUSY scattering properties of dielectric coupled waveguides have also been observed \cite{Heinrich2014a}.

As we have described so far, many studies have been done for the electromagnetic SUSY, but their focusing point is mainly limited to dielectric structures.
Recent progress of plasmonics \cite{AlexanderMaier2007} and metamaterials \cite{Solymar2009} using metals in optics demands further studies of SUSY for metallic systems.
To design and analyze the characteristics of metallic structures, intuitive electrical circuit models are very useful, because they extract the nature of the phenomena despite reducing the degree of freedom for the problem \cite{Nakata2012a}.
Actually, a circuit-theoretical design strategy called {\it metactronics} has been proposed even in the optical region \cite{Engheta2007} and the circuit theory for plasmons has also been developed \cite{Staffaroni2012}.
If we could design circuit models enabling exotic phenomena, they open up new possibilities for application to higher frequency ranges due to the scale invariance of Maxwell equations.
Thus, in this paper we develop how SUSY appears in inductor-capacitor circuit networks and demonstrate the SUSY correspondence in the high frequency region.
In particular, we focus on the SUSY structure for inductor-capacitor circuit networks on a graph and its line graph.

This article is organized as follows. 
In Sec.~\ref{sec:2}, we start by introducing the graph-theoretical concepts and formulate a general class of inductor-capacitor circuit network pairs related through SUSY, derived from the topological nature of the graphs representing the circuits. 
In Sec.~\ref{sec:3}, we theoretically and experimentally demonstrate the SUSY eigenfrequency correspondence for paired metallic lattices in the terahertz frequency range.
In Sec.~\ref{sec:4}, we summarize and conclude the paper.

\section{Theory \label{sec:2}}
\subsection{Eigenequation for inductor-capacitor circuit networks }

We consider an inductor-capacitor circuit network on a simple directed graph $G=(\mathcal{V},\mathcal{E})$, where $\mathcal{V}$ and $\mathcal{E}$ are the sets of vertices and directed edges, respectively.
The modifier {\it simple} means that there are no multiple edges between any vertex pair and no edge (loop) that connects a vertex to itself.
The number of the edges connected to a vertex $v$ of 
a graph is called the degree of $v$.
A {\it regular} graph is a graph whose
every vertex has the same degree.
We assume that $G$ is an $m$-{\it regular} graph with all vertices having degree $m$.
The capacitors, all with the same capacitance $C$, are connected between each vertex $v\in \mathcal{V}$ and the ground.
Coils, all with the same inductance $L$, are loaded along all $e \in \mathcal{E}$.
An example of $G$ and the inductor-capacitor circuit network on it are shown in Fig.~\ref{fig:lc_ladder}(a) and (b).
\begin{figure}[bp]
\includegraphics[width=86mm]{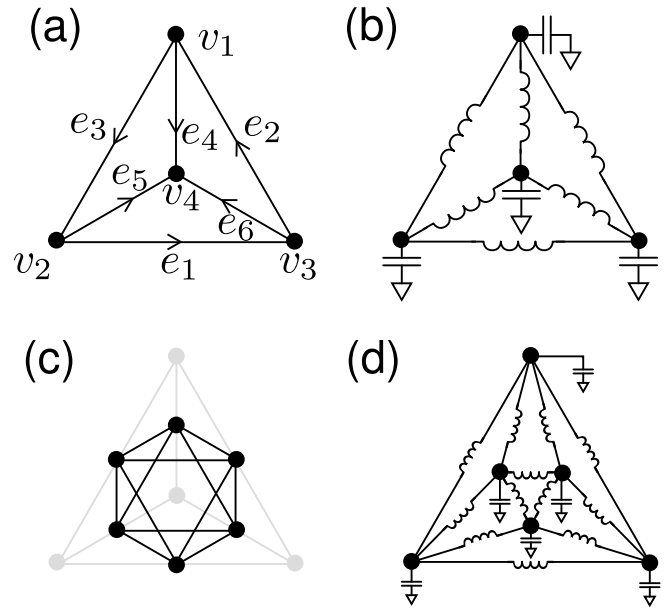}
\caption{\label{fig:lc_ladder} (a) Example of simple $3$-regular directed graph. (b) Inductor-capacitor circuit network on the graph. (c) Line graph of the graph shown in (a). (d) Inductor-capacitor circuit network on the line graph (c).
}
\end{figure}
For $v\in \mathcal{V}$ and $e\in \mathcal{E}$, the incidence matrix $\mathsf{X}=[X_{ve}]$ of a directed graph $G$ is defined as follows:
$X_{ve}=-1$ ($e$ enters $v$), $X_{ve}=1$ ($e$ leaves $v$), otherwise $X_{ve}=0$.

Using vector notation, we represent the current distribution $J_e$ flowing along $e\in \mathcal{E}$ as a column vector $\vct{J}=[J_e]^\mathrm{T}$.
The charge distribution is denoted by $\vct{q}=[q_v]^\mathrm{T}$ with a stored charge $q_v$ at $v\in \mathcal{V}$.
The charge conservation law is given by
\begin{equation}
\dot{\vct{q}}=-\mathsf{X} \vct{J},  \label{eq:1}
\end{equation}
where the time derivative is represented by the dot.
The scalar potential $\Phi_v$ at $v \in {\mathcal V}$ must satisfy Faraday's law of induction, so we have
\begin{equation}
\dot{\vct{J}}=\frac{1}{L} \mathsf{X}^\mathrm{T}\vct{\Phi} ,  \label{eq:2}
\end{equation}
with $\vct{\Phi}=[\Phi_v]^\mathrm{T}$.
The scalar potential $\vct{\Phi}$ is written as
\begin{equation}
 \vct{\Phi}=\mathsf{P} \vct{q},  \label{eq:3}
\end{equation}
with a potential matrix $\mathsf{P}$.
In our case, $\mathsf{P}$ is given by
\begin{equation}
 \mathsf{P}=C^{-1}\mathsf{I},  \label{eq:4}
\end{equation}
where we use the identity matrix $\mathsf{I}$.

From Eqs.~(\ref{eq:1})--(\ref{eq:4}), we obtain
\begin{equation}
 \ddot{\vct{q}}=-{\omega_0}^2\mathsf{X}\mathsf{X}^\mathrm{T} \vct{q}, \nonumber
\end{equation}
with $\omega_0=1/\sqrt{LC}$.
Assuming $\vct{q}=\tilde{\vct{q}}\exp(-\ii \omega t)+\cc$, we have an eigenequation
\begin{equation}
 \mathsf{L}\tilde{\vct{q}}= \left(\frac{\omega}{\omega_0}\right)^2\tilde{\vct{q}}  \label{eq:5}
\end{equation}
with the Laplacian $\mathsf{L}=\mathsf{X}\mathsf{X}^\mathrm{T}$.
We introduce an adjacency matrix $\mathsf{A}=[A_{vw}]$, where $A_{vw}$ is 1 if $v, w\in \mathcal{V}$ are connected by an edge, otherwise 0.
From a direct calculation, we can write $\mathsf{L}$ by $\mathsf{A}$ as 
\begin{equation}
\mathsf{L}=\mathsf{X}\mathsf{X}^\mathrm{T}=-\mathsf{A}+m\mathsf{I},  \label{eq:6}
\end{equation}
where $m$ is the degree of the vertex of $G$.
Note that $\mathsf{L}$ is independent of the direction of the edges in $G$ because $\mathsf{L}$ is expressed in terms of $\mathsf{A}$ and $\mathsf{I}$.

The directed graph $G$ can be also regarded as an undirected graph.
For $e\in \mathcal{E}$, we can make an undirected edge $\bar{e}$, where the bar operator ignores the direction of the edge.
Then, we have $\bar{G}=(\mathcal{V},\bar{\mathcal{E}})$ with $\bar{\mathcal{E}}=\{\bar{e}|e\in \mathcal{E}\}$.
We can also define the undirected incidence matrix $\bar{\mathsf{X}}=[\bar{X}_{ve}]$ as $\bar{X}_{ve}=1$ ($e$ and $v$ are connected), otherwise $\bar{X}_{ve}=0$.
Using $\bar{\mathsf{X}}$, $\mathsf{A}$ is written as follows \cite{Biggs1994}:
\begin{equation}
 \bar{\mathsf{X}}\bar{\mathsf{X}}^\mathrm{T}=\mathsf{A}+m \mathsf{I}.  \label{eq:7}
\end{equation}
From Eqs.~(\ref{eq:6}) and (\ref{eq:7}), we obtain
\begin{equation}
 \mathsf{L}=- \bar{\mathsf{X}}\bar{\mathsf{X}}^\mathrm{T}+2m\mathsf{I}.  \label{eq:8}
\end{equation}

\subsection{SUSY correspondence in spectra on a simple regular graph and its line graph}
Next, we introduce the line graph concept \cite{Biggs1994}.
The line graph $L(G)=(\mathcal{V}\sub{L}, \mathcal{E}\sub{L})$ of a directed graph $G$ is constructed as follows. 
Each edge in $G$ is considered to be a vertex of $L(G)$. 
Two vertices of $L(G)$ are connected if the corresponding edges in $G$ have a vertex in common.
There are two possible choices for the direction of each edge in $L(G)$ and we adapt one of them. 
From here on, we only consider $L(G)$ of a simple $m$-regular graph $G$.
In this case, the line graph $L(G)$ is a simple $m\sub{L}$-regular graph.
The degree $m\sub{L}$ can be represented by $m$.
For a vertex $v \in \mathcal{V}$ included in $e\in \mathcal{E}$, there are $m-1$ edges $e'\in \mathcal{E}\setminus\{e \}$ connected to $v$.
Then, we obtain
\begin{equation}
 m\sub{L}=2(m-1).  \label{eq:9}
\end{equation}
Note that $(m/2)\#\mathcal{V}= \#\mathcal{E}=\#\mathcal{V}\sub{L}=(2/m\sub{L})\#\mathcal{E}\sub{L}$ is satisfied for a finite graph $G$, where $\#\mathcal{S}$ represents the numbers of the elements of the set $\mathcal{S}$.
Figure~\ref{fig:lc_ladder}(c) is an example of the line graph of the graph $G$ shown in Fig.~\ref{fig:lc_ladder}(a).
Figure~\ref{fig:lc_ladder}(d) is the inductor-capacitor circuit network on $L(G)$.

In the context of mathematics, it is known that the spectra of the Laplacians for the graph and its line graph are related to each other
\cite{Shirai2000}.  
For the convenience of the readers, we rederive this property in a simple manner and apply it to the inductor-capacitor circuit networks.
The Laplacian of $L(G)$ is written as
\begin{equation}
 \mathsf{L}\sub{L}=\mathsf{X}\sub{L}{\mathsf{X}\sub{L}}^\mathrm{T}=-\mathsf{A}\sub{L}+m\sub{L}\mathsf{I}\sub{L},  \label{eq:10}
\end{equation}
with the identity matrix $\mathsf{I}\sub{L}$, the incidence matrix $\mathsf{X}\sub{L}$, and the adjacency matrix $\mathsf{A}\sub{L}$ for $L(G)$. 
The adjacency matrix of $L(G)$ is represented as follows \cite{Biggs1994}:
\begin{equation}
 \mathsf{A}\sub{L}=\bar{\mathsf{X}}^\mathrm{T} \bar{\mathsf{X}}-2\mathsf{I}\sub{L}.  \label{eq:11}
\end{equation}
From Eqs.~(\ref{eq:10}) and (\ref{eq:11}), we have
\begin{equation}
 \mathsf{L}\sub{L}=-\bar{\mathsf{X}}^\mathrm{T} \bar{\mathsf{X}}+(m\sub{L}+2)\mathsf{I}\sub{L}.
  \label{eq:12}
\end{equation}
Now, we consider the composite system of $L(G)$ and $G$.
Then the composite Laplacian $\mathcal{L}\sub{c}$ is given by
\begin{equation}
\mathcal{L}\sub{c}=-\mathcal{K}\sub{c} +2m\mathcal{I}\sub{c},  \label{eq:13}
\end{equation}
with 
\begin{equation}
  \mathcal{K}\sub{c}=
    \begin{bmatrix}
   \bar{\mathsf{X}}^\mathrm{T}\bar{\mathsf{X}}&0\\
   0&\bar{\mathsf{X}}\bar{\mathsf{X}}^\mathrm{T}
  \end{bmatrix},\
  \mathcal{I}\sub{c}=
    \begin{bmatrix}
   \mathsf{I}\sub{L}&0\\
   0&\mathsf{I}
  \end{bmatrix},
\nonumber
\end{equation}
where we have used Eqs.~(\ref{eq:8}), (\ref{eq:9}), and (\ref{eq:12}).
The composite operator $\mathcal{K}\sub{c}$ is
written as
\begin{equation}
\mathcal{K}\sub{c}=\mathcal{Q}\mathcal{Q}^\dagger+\mathcal{Q}^\dagger\mathcal{Q},  \label{eq:14}
\end{equation}
where the symbol $\dagger$ represents the Hermitian conjugate, and we define the supercharge as 
$$
\mathcal{Q} =
\begin{bmatrix}
 0 & 0\\
 \bar{\mathsf{X}} & 0
\end{bmatrix}.
$$
These operators satisfy the superalgebra \cite{Cooper1994}:
$$[\mathcal{K}\sub{c},\mathcal{Q}]=[\mathcal{K}\sub{c},\mathcal{Q}^\dagger]=0,$$
$$\{\mathcal{Q},\mathcal{Q}^\dagger\}=\mathcal{K}\sub{c},\ \{\mathcal{Q},\mathcal{Q}\}=\{\mathcal{Q}^\dagger,\mathcal{Q}^\dagger\}=0,
$$
where $\{\mathcal{A},\mathcal{B}\}$ and $[\mathcal{A},\mathcal{B}]$
are the anticommutator and the commutator, respectively.
Therefore, the eigenspectra of the inductor-capacitor circuit networks on the simple regular graph and its line graph must coincide except, possibly, for the highest eigenfrequency. 
Actually, if we have eigenvector $\vct{x}$ satisfying $\mathsf{L}\vct{x}=E\vct{x}$ with the eigenvalue $E$, we obtain $\bar{\mathsf{X}}^\mathrm{T}\vct{x}$, satisfying $\mathsf{L}\sub{L}(\bar{\mathsf{X}}^\mathrm{T}\vct{x})=E(\bar{\mathsf{X}}^\mathrm{T}\vct{x})$.
Then, we have an eigenvector $\bar{\mathsf{X}}^\mathrm{T}\vct{x}$  for $\mathsf{L}\sub{L}$ when 
$\bar{\mathsf{X}}^\mathrm{T}\vct{x}\ne \vct{0}$ ($E\ne 2m$).
The eigenvalue $E$ of $\mathsf{L}$ and $\mathsf{L}\sub{L}$
must satisfy $E\leq 2m$, because
$\bar{\mathsf{X}}\bar{\mathsf{X}}^\mathrm{T}$ and $\bar{\mathsf{X}}^\mathrm{T}\bar{\mathsf{X}}$ are positive-semidefinite.
For an eigenvector $\vct{y}_{2m}$ of $\mathsf{L}$ ($\mathsf{L}\sub{L}$) with eigenvalue $E=2m$,
the partner mode cannot be obtained by multiplying $\bar{\mathsf{X}}^\mathrm{T}$ ($\bar{\mathsf{X}}$), because of $\bar{\mathsf{X}}^\mathrm{T}\vct{y}_{2m}=0$ ($\bar{\mathsf{X}}\vct{y}_{2m}=0$).
The condition for complete spectral coincidence of the eigenvalues of $\mathsf{L}$ and $\mathsf{L}\sub{L}$ is discussed in Appendix~\ref{sec:appA}.
Note that quantum tight-binding models represented by Eqs.~(\ref{eq:8}) and (\ref{eq:12}) are isospectral except, possibly, for the highest eigenenergy, but an accurate tuning of the on-site potential satisfying Eqs.~(\ref{eq:8}) and (\ref{eq:12}) is usually difficult to achieve. 
The significant point of SUSY for the inductor-capacitor circuit networks is that the on-site potential tuning is accomplished naturally.

If $G$ is a periodic graph with lattice vectors $\{\thrvct{a}_i\}$, we can show that
the spectral coincidence except possibly for the highest eigenfrequency holds for {\it each wave vector. }
We define parallel translation with $\thrvct{a}_i$ as
$\mathsf{P}^\mathcal{E}_i$ and $\mathsf{P}^\mathcal{V}_i$
for edges and vertices, respectively.
From the translational symmetry,
we have $\bar{\mathsf{X}}\mathsf{P}^\mathcal{E}_i=\mathsf{P}^\mathcal{V}_i\bar{\mathsf{X}}$. For a Bloch vector $\vct{x}_{\thrvct{k}}$ satisfying $\mathsf{P}^\mathcal{E}_i \vct{x}_{\thrvct{k}}=\exp({\ii\thrvct{k}\cdot\thrvct{a}_i}) \vct{x}_{\thrvct{k}}$, we have $\mathsf{P}^\mathcal{V}_i(\bar{\mathsf{X}}\vct{x}_{\thrvct{k}})=\exp(\ii\thrvct{k}\cdot\thrvct{a}_i)(\bar{\mathsf{X}}\vct{x}_{\thrvct{k}})$. This means that
$\bar{\mathsf{X}}\vct{x}_{\thrvct{k}}$ is also a Bloch vector. Similar discussion can be applied to
$\bar{\mathsf{X}}^\mathrm{T}$. Then, $\bar{\mathsf{X}}$ and
$\bar{\mathsf{X}}^\mathrm{T}$ map Bloch vectors to Bloch vectors without changing $\thrvct{k}$.
Therefore, the decomposition shown in Eq.~(\ref{eq:13}) is valid in
the subspace of Bloch vectors with a wave vector $\thrvct{k}$.
This means that the spectral coincidence except, possibly, for the highest eigenfrequency
holds for each wave vector.

 \subsection{Examples \label{subsec:c}}
\subsubsection{finite case\label{subsub:1}}
For the graph shown in Fig.~\ref{fig:lc_ladder}(a), we have
\begin{equation}
\bar{\mathsf{X}}=\begin{bmatrix}
  0 &1 &1 &1 &0 &0\\
  1 &0 &1 &0 &1 &0\\
  1 &1 &0 &0 &0 &1\\  
  0 &0 &0 &1 &1 &1
		 \end{bmatrix}\nonumber.  
\end{equation}
Then, we get $\omega=2\omega_0,\ 2\omega_0,\ 2\omega_0,\ 0$ for the inductor-capacitor circuit network on $G$.
On the other hand, $\omega=\sqrt{6}\omega_0,\ \sqrt{6}\omega_0,\ 2\omega_0,\ 2\omega_0,\ 2\omega_0,\ 0$ are obtained for the inductor-capacitor circuit network on $L(G)$.
We can see that all angular eigenfrequencies for $G$ are included in those for $L(G)$.

\subsubsection{Infinite case}

\begin{figure}[!t]
\includegraphics[width=86mm]{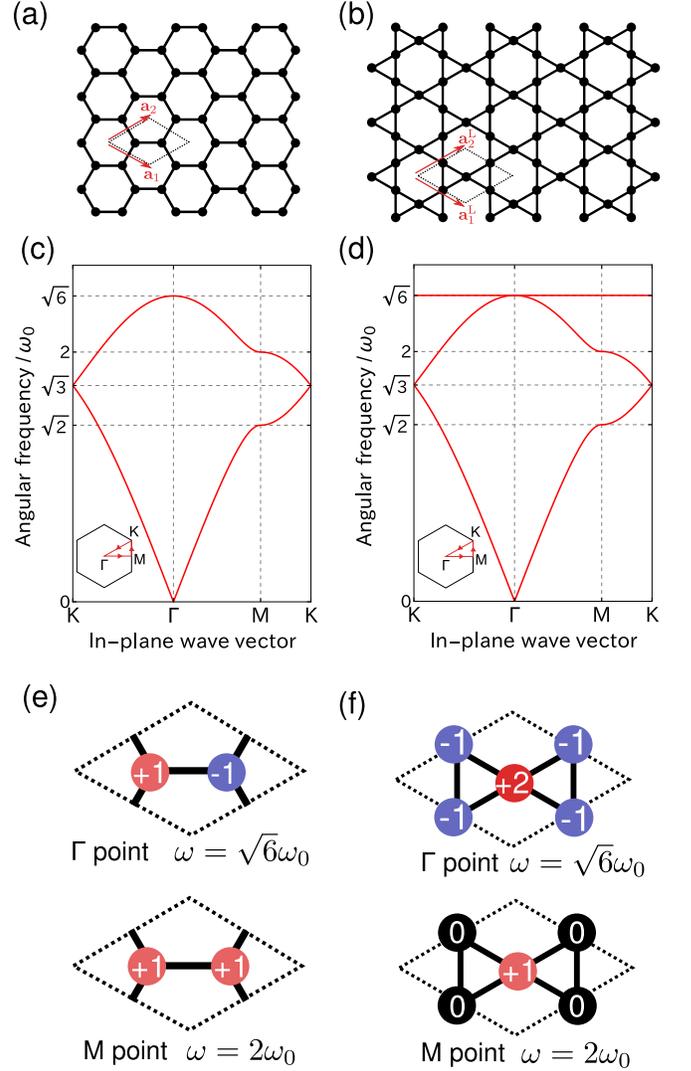}
\caption{\label{fig:kagome_hex}  (color online) (a) Hexagonal lattice. (b) \Kagome lattice. (c) Dispersion relation for a hexagonal inductor-capacitor circuit network. (d) Dispersion relation for a \kagome inductor-capacitor circuit network. (e) Eigenmodes of the higher band at the $\Gamma$ and $\mathrm{M}$ points for the hexagonal lattice. (f) Eigenmodes of the middle band at the $\Gamma$ and $\mathrm{M}$ points for the \kagome lattice.
 }
\end{figure}

Here, we consider a hexagonal lattice as $G$ [Fig.~\ref{fig:kagome_hex}(a)].
The line graph $L(G)$ is a \kagome lattice [Fig.~\ref{fig:kagome_hex}(b)].
To see the spectra coincidence directly, we calculate the angular eigenfrequencies $\omega$.
At first, we calculate the eigenvalue $\alpha$ and $\alpha\sub{L}$ for $\mathsf{A}$ and $\mathsf{A}\sub{L}$, respectively.
Due to the Bloch theorem, it is enough to calculate them in the restricted space $\mathcal{W}_{\thrvct{k}_\parallel}$ of waves with wave vector $\thrvct{k}_\parallel$.
For the hexagonal lattice, we have two vertices $v_p\in \mathcal{V}$ $ (p=1,2)$ in a unit cell.
The vertex displaced from $v_p$, with $i\thrvct{a}_1+j\thrvct{a}_2$, is denoted by $v_p^{(i,j)}$ for $(i,j)\in \mathbb{Z}^2$, where $\thrvct{a}_1$ and $\thrvct{a}_2$ are lattice vectors for $G$.
Now, we define 
$\vct{\Psi}_1 (\thrvct{k}_\parallel)=\sum_{(i,j)\in\mathbb{Z}^2} \exp(\ii \thrvct{k}_\parallel\cdot(i\thrvct{a}_1+j\thrvct{a}_2) ){\vct{v}_1^{(i,j)}}$ and $\vct{\Psi}_2 (\thrvct{k}_\parallel)=\sum_{(i,j)\in\mathbb{Z}^2} \exp(\ii \thrvct{k}_\parallel\cdot(i\thrvct{a}_1+j\thrvct{a}_2) ){\vct{v}_2^{(i,j)}}$, where $\{\vct{v}_{p}^{(i,j)}|(i,j)\in\mathbb{Z}^2, p=1,2\}\subset \mathcal{H}$ is a complete orthogonal basis of the Hilbert space $\mathcal{H}$. 
The vector subspace $\mathcal{W}_{\thrvct{k}_\parallel}$ is spanned by $\vct{\Psi}_1 (\thrvct{k}_\parallel)$ and $\vct{\Psi}_2 (\thrvct{k}_\parallel)$.
The  action of $\mathsf{A}$ in the restricted space $\mathcal{W}_{\thrvct{k}_\parallel}$ is represented by a $2\times 2$ matrix $\mathsf{A}({\thrvct{k}_\parallel})=[A_{ij}({\thrvct{k}_\parallel})]$,  satisfying
$\mathsf{A}\vct{\Psi}_i({\thrvct{k}_\parallel})= \sum_{j=1}^{2}\vct{\Psi}_j({\thrvct{k}_\parallel}) A_{ji}({\thrvct{k}_\parallel})$.
Diagonalizing $\mathsf{A}({\thrvct{k}_\parallel})$ we have
\begin{equation}
\alpha(\thrvct{k}_\parallel)=\pm\sqrt{3+2F(\thrvct{a}_1, \thrvct{a}_2;\thrvct{k}_\parallel)},  \label{eq:15}
\end{equation}
with $F(\thrvct{u}_1, \thrvct{u}_2; \thrvct{k}_\parallel)=\cos (\thrvct{k}_\parallel\cdot\thrvct{u}_1)+\cos( \thrvct{k}_\parallel\cdot\thrvct{u}_2)+\cos \thrvct{k}_\parallel\cdot(\thrvct{u}_1-\thrvct{u}_2)$.
By applying a similar calculation to the \kagome lattice, we obtain
\begin{equation}
\alpha\sub{L}(\thrvct{k}_\parallel) =-2, 1\pm\sqrt{3+2F(\thrvct{a}_1\sur{L}, \thrvct{a}_2\sur{L};\thrvct{k}_\parallel)},  \label{eq:16}
\end{equation}
with lattice vectors $\thrvct{a}_1\sur{L}$ and $\thrvct{a}_2\sur{L}$ for $L(G)$.
Using Eqs.~(\ref{eq:5}), (\ref{eq:6}), and (\ref{eq:15}), we obtain
\begin{equation}
\frac{\omega}{\omega_0}=\sqrt{3\pm\sqrt{3+2F(\thrvct{a}_1, \thrvct{a}_2;\thrvct{k}_\parallel)}}  \label{eq:17}
\end{equation}
for the hexagonal lattice.
From Eqs.~(\ref{eq:5}), (\ref{eq:10}), and (\ref{eq:16}), the \kagome lattice also has the dispersion relation
\begin{equation}
\frac{\omega}{\omega_0}=\sqrt{6},  \sqrt{3\pm\sqrt{3+2F(\thrvct{a}_1\sur{L}, \thrvct{a}_2\sur{L};\thrvct{k}_\parallel)}}.  \label{eq:18}
\end{equation}
The obtained dispersion relations are shown in Figs.~\ref{fig:kagome_hex}(c) and (d).
The lower two bands are identical as we expected.
Note that these bands are determined only by the product $LC$ and are independent of the ratio $L/C$.
We also show examples of the eigenmodes for the hexagonal and \kagome lattices in Figs.~\ref{fig:kagome_hex}(e) and (f).

\section{Band correspondence between metallic hexagonal and \kagome lattices \label{sec:3}}

In the previous section, we developed inductor-capacitor circuit networks that are related through SUSY. 
As an example, we saw that the bands of hexagonal and \kagome inductor-capacitor circuit networks are isospectral by SUSY, except for the highest band of the \kagome lattice.
In this section, we examine this correspondence for realistic system.
It is known that bar-disk resonators composed of metallic disks connected by metallic bars can be qualitatively modeled by the inductor-capacitor circuit networks discussed in Sec.~\ref{sec:2}, because charges on the disks are coupled dominantly by the current flowing along the bars \cite{Nakata2012,Kajiwara2016}.
The modes on the bar-disk resonators are called spoof plasmons.
Here, we study the spoof plasmons of metallic hexagonal and \kagome lattices whose designs are shown in Fig.~\ref{fig:kagome_hex_lattice_design}.

\begin{figure}[!tp]
 \includegraphics[width=86mm]{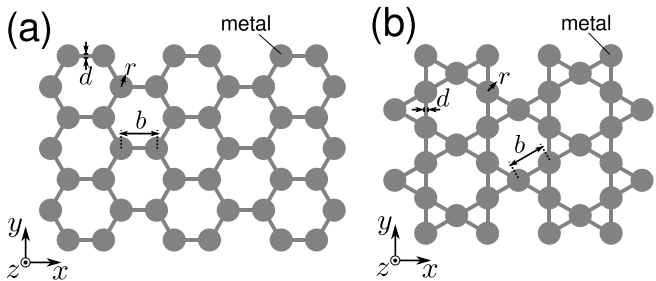}
\caption{\label{fig:kagome_hex_lattice_design} Designs for (a) metallic hexagonal lattice and (b) metallic \kagome lattice. 
	The following parameters are used: $d=10\,\U{\mu m}$, $r=150\,\U{\mu m}$, $b=800/\sqrt{3}\,\U{\mu m}$, and thickness $h=30\,\U{\mu m}$.}
\end{figure}

\subsection{Simulation\label{sec:3-1}}

\begin{figure}[btp]
\includegraphics[width=86mm]{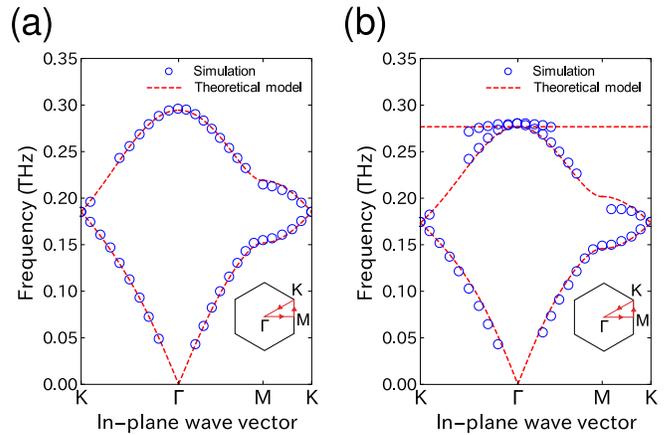}
\caption{\label{fig:hex_kagome_bands} (color online) Dispersion relations obtained by simulation for a 
(a) metallic hexagonal lattice and
(b) metallic \kagome lattice.
The fitting parameters for the theoretical models discussed in Appendix~\ref{sec:appB} are given by $\omega_0\sur{hex}=2\pi \times 0.107\,\U{THz}$, $\eta\sur{hex}_0=0.0916$ for the hexagonal lattice and $\omega_0\sur{kag}=2\pi\times 0.101\,\U{THz}$ and $\eta\sur{kag}_0=0.142$ for the \kagome lattice. 
}
\end{figure}
\begin{figure*}[hbtp]
\includegraphics[width=172mm]{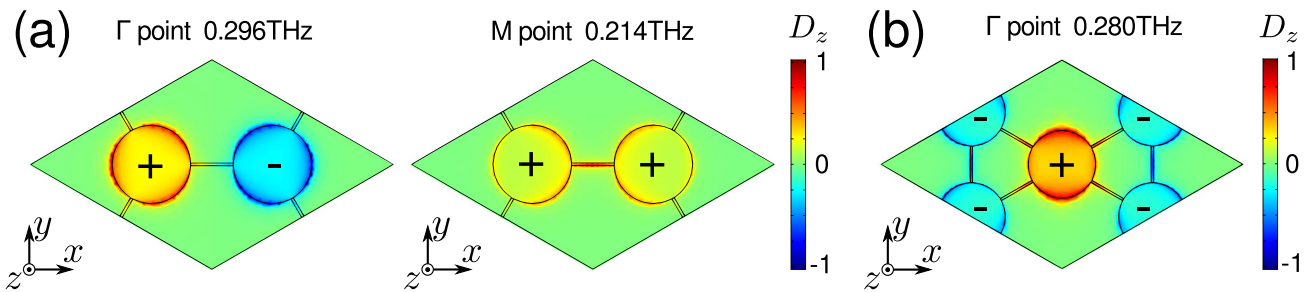}
\caption{\label{fig:eigenmodes} (color online) Electric flux density $D_z$ on $z=h/2$:
(a) eigenmodes of the higher band at the $\Gamma$ and $\mathrm{M}$ points for the metallic hexagonal lattice and (b) eigenmode of the middle band at the $\Gamma$ point for the metallic \kagome lattice. 
}
\end{figure*}

We perform an eigenfrequency analysis for the metallic hexagonal and \kagome lattices by the finite element method solver (\textsc{Comsol Multiphysics}).
The parameters of the structures for Fig.~\ref{fig:kagome_hex_lattice_design} are as follows:
bar width $d=10\,\U{\mu m}$, radius of disks $r=150\,\U{\mu m}$, distance between nearest disks $b=800/\sqrt{3}\,\U{\mu m}$, and thickness $h=30\,\U{\mu m}$.
In each simulation, the finite thickness metallic lattice parallel to $z=0$ is located in $z\in [-h/2,h/2]$.
The unit cell in the $xy$ plane is the rhombus spanned by the lattice vectors and denoted by $U$.
To reduce the degrees of freedom, we use the mirror symmetry with respect to $z=0$. 
A simulation domain with the material parameters of a vacuum is set in $U\times [0,6l]$ with $l=\sqrt{3} b=800\,\U{\mu m}$, and a perfect magnetic conductor condition imposed on the surface $z=0$.
Half of the structure in $z\in [0,h/2]$ is engraved in the simulation domain and a perfect electric  conductor (PEC) boundary condition is imposed on the structure surface.
A perfect matched layer (PML) in $U\times [5l, 6l]$ with a PEC boundary at $z=6l$ is used to truncate the infinite effect.
The periodic boundary condition with a phase shift (Floquet boundary conditions with a wave vector $\thrvct{k}_\parallel$)  is applied to $\partial U \times [0,6l]$.
Changing $\thrvct{k}_\parallel$ along the Brillouin zone boundary, we calculate the eigenfrequencies.
To remove the  modes which are not localized near the metallic surface \cite{Parisi2012}, we select the modes with 
\begin{equation}
 \xi=\frac{\int_{U\times [9l/2, 6l]} |\tilde{\thrvct{E}}|^2 \dd V}{\int_{U\times [h/2, l/2]} |\tilde{\thrvct{E}}|^2 \dd V}<1,\nonumber 
\end{equation}
where the complex amplitude of the electric field of the mode is denoted by $\tilde{\thrvct{E}}$.

The calculated eigenfrequencies for the metallic hexagonal and \kagome lattices are shown in Fig.~\ref{fig:hex_kagome_bands} as circles.
Note that some points are missing because unphysical modes located near PML accidentally exist or couple with the modes. 
As we explained earlier, we eliminated such modes with $\xi\geq 1$.
In Fig.~\ref{fig:hex_kagome_bands}, we observe the lower two band correspondence between the metallic hexagonal and \kagome lattices.
The bands of the metallic hexagonal lattice are about 5\% higher than those of the \kagome lattice.
However, we can say that the band correspondence is qualitatively established.
The detailed theoretical models for fitting curves are discussed later in Appendix.~\ref{sec:appB}.
Figure~\ref{fig:eigenmodes} shows the electric flux density $D_z$ on $z=h/2$ of the  specific modes, where $D_z$ corresponds to the surface charge on the metal.
These mode profiles agree with the theoretically calculated eigenmodes shown in Figs.~\ref{fig:kagome_hex}(e) and (f).

\begin{figure}[!b]
 \includegraphics[width=86mm]{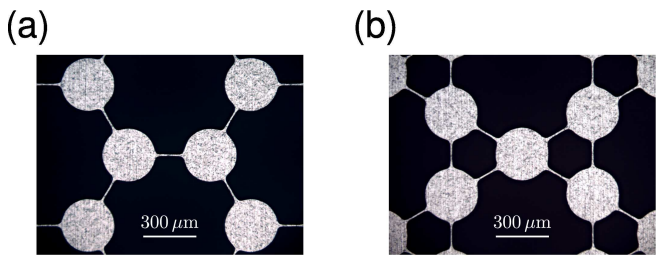}
\caption{\label{fig:hdbr_kdbr} (color online) Microphotographs of the (a) metallic hexagonal and (b) metallic \kagome lattices.}
\end{figure}

\subsection{Experiment}

\begin{figure*}[tbp]
\includegraphics[width=172mm]{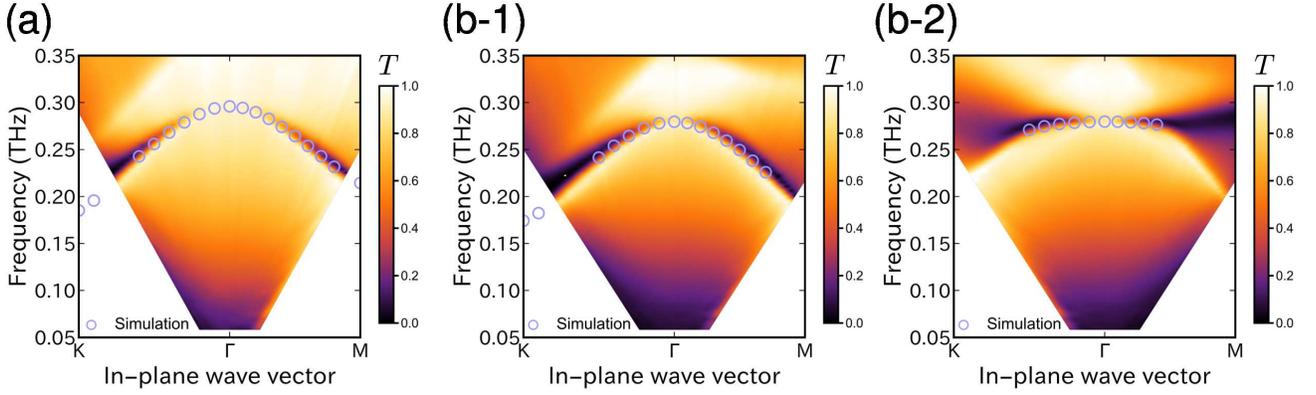}
\caption{\label{fig:transmission} (color online) Power transmission spectra mapped on the wave-vector-frequency plane for a (a) metallic hexagonal lattice and (b-1, b-2) metallic \kagome lattice.
For (a) and (b-1), the incident waves are set as transverse electric modes in the $\Gamma$--$\mathrm{K}$ scan and transverse magnetic modes in the $\Gamma$--$\mathrm{M}$ scan.
On the other hand, in (b-2), we use transverse magnetic excitation in the $\Gamma$--$\mathrm{K}$ scan and transverse electric excitation for the $\Gamma$--$\mathrm{M}$ scan.
The calculated eigenfrequencies are depicted as circles.}
\end{figure*}

To investigate the dispersion relation experimentally, we fabricated the metallic hexagonal and \kagome lattices by etching and performed transmission measurement on them using the terahertz time-domain spectroscopy technique.
The samples made of stainless steel (SUS304) are shown in Fig.~\ref{fig:hdbr_kdbr}.
The geometrical parameters of these samples are the same as those for the simulation model in Sec.~\ref{sec:3-1}.
The area where structures are patterned is $4\,\U{cm}\times 4\,\U{cm}$.
The terahertz beam is generated by a spiral antenna and collimated by a combination of a hyper-hemispherical silicon lens and a 
Tsurupica$^\text{\textregistered}$ lens.
The beam diameter is set to $13\, \U{mm}$ by an aperture.
Wire-grid polarizers are located near the emitter and detector, which are adjusted so that the emitted and detected fields have the same linear polarization.
The transmission spectrum $T(\omega)$ in the frequency domain is obtained from $T(\omega)=|\tilde{E}(\omega)/\tilde{E}\sub{ref}(\omega)|^2$, where
$\tilde{E}(\omega)$ and $\tilde{E}\sub{ref}(\omega)$ are Fourier transformed electric fields with and without the sample.
To scan the Brillouin zone, power transmission spectra are measured with changing incident angle $\theta$ from $\theta=0^\circ$ to $60^\circ$ with step $2.5^\circ$.
Here, the magnitude of the in-plane wave vector $\thrvct{k}_\parallel$ is given by $k_\parallel = (\omega/c) \sin \theta$, where $c$ is the speed of light.

To observe the higher band of the hexagonal lattice and the middle band of the \kagome lattice, the incident waves are set as follows:
(i) transverse electric (TE) modes in $\Gamma$--$\mathrm{K}$ scan and
(ii) transverse magnetic (TM) modes in $\Gamma$--$\mathrm{M}$ scan.
Figures~\ref{fig:transmission}(a) and (b-1) show the power transmission spectra for these incident waves entering into the metallic hexagonal and \kagome lattices, respectively.
The calculated eigenfrequencies are shown simultaneously as circles in Fig.~\ref{fig:transmission}.
We can see that the transmission dips form a band from $0.15$ to $0.3\,\U{THz}$.
The calculated eigenfrequencies are located around the experimental transmission dips.
Thus, the SUSY band correspondence for the second band is experimentally demonstrated.

The highest band for the metallic \kagome lattice can be observed for differently polarized incident waves.
Figure~\ref{fig:transmission}(b-2) shows the transmission spectra for the metallic \kagome lattice, where the incident waves are set as 
(i) TM modes in the $\Gamma$--$\mathrm{K}$ scan and 
(ii) TE modes in the $\Gamma$--$\mathrm{M}$ scan. 
In Fig.~\ref{fig:transmission}(b-2), we can see the flat band reported in Ref.~\onlinecite{Nakata2012}.
Note that the frequencies of the lowest band modes are under the light line, so it is impossible to excite them by free-space plane waves.
To excite the lowest band modes, another method, e.g., attenuated total reflection measurement, is needed \cite{Kajiwara2016}.

  \subsection{Discussion \label{sec:3.3}}
\begin{figure}[tp]
\includegraphics[width=86mm]{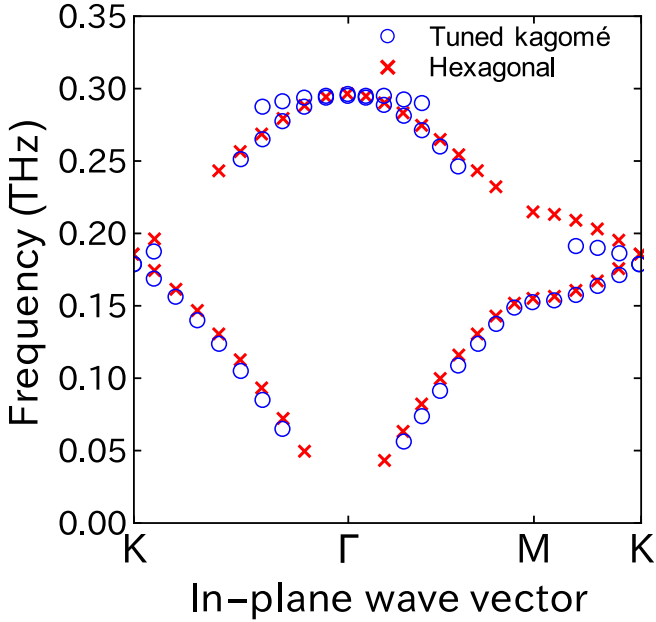}
\caption{\label{fig:tuned_kagome} (color online)
Comparison between eigenfrequencies for a tuned metallic \kagome lattice with $d=28\,\U{\mu m}$ and metallic hexagonal lattice with $d=10\,\U{\mu m}$. The other parameters are $r=150\,\U{\mu m}$, $b=800/\sqrt{3}\,\U{\mu m}$, and thickness $h=30\,\U{\mu m}$.
 }
\end{figure}

In the previous subsections,
the band correspondence between the metallic hexagonal and \kagome lattices was
demonstrated, but a $5\%$ discrepancy between the bands was also observed.
Here, we investigate the possibility to compensate empirically for the discrepancy.
We calculated the eigenfrequencies for the metallic \kagome lattice with the bar width $d=28\,\U{\mu m}$.
The other parameters are the same as the previous one.
The calculated results are shown in Fig.~\ref{fig:tuned_kagome} as circles.
To compare with the previous result, the eigenfrequencies for the metallic hexagonal lattice with $d=10\,\U{\mu m}$ are also plotted in Fig.~\ref{fig:tuned_kagome}.
We can see the improvement of band correspondence between the metallic hexagonal lattice and tuned metallic \kagome lattice.

 \section{Conclusion \label{sec:4}}

In this paper, we showed that the inductor-capacitor circuit networks on a simple regular graph and its line graph are related through SUSY, and their spectra must coincide (except possibly for the highest eigenfrequency).
The SUSY structure for the circuits was derived from the topological nature of the graphs.
To observe SUSY correspondence of the bands in the high frequency range, we investigated the metallic hexagonal and \kagome lattices.
The band correspondence between them was predicted by a simulation.
We performed terahertz time-domain spectroscopy for these metallic lattices and observed the band correspondence.
Finally, we proposed an empirical tuning method to reduce
the discrepancy of the corresponding bands of the metallic hexagonal and \kagome lattices.

The theoretical results are formulated  for the inductor-capacitor circuit networks
and independent of the implementations.
Therefore, our results is also applicable to transmission-line systems such as microstrip.
The SUSY correspondence in the spectra of inductor-capacitor circuit networks has the potential to extend mode filters \cite{Miri2013} and 
mode converters \cite{Heinrich2014} in two-dimensional (spoof) plasmonic systems.

\begin{acknowledgments}
The present research was supported by the JSPS KAKENHI Grant No. 25790065 and Grant-in-Aid for JSPS Fellows No. 13J04927.
Two of the authors (Y.~N. and Y.~U.) were supported by JSPS Research Fellowships for Young Scientists.
\end{acknowledgments}

\appendix

 \section{Condition for complete spectral coincidence of eigenvalues of $\mathsf{L}$ and $\mathsf{L}\sub{L}$\label{sec:appA}}
In this Appendix, we consider condition for complete spectral coincidence of the eigenvalues of $\mathsf{L}$ and $\mathsf{L}\sub{L}$.
Here, we mainly consider the finite graph cases.
At first, we introduce a bipartite graph.
A bipartite graph is a graph whose
vertex set can be separated into
two disjoint sets $\mathcal{V}_1$ and $\mathcal{V}_2$
such that every edge is connected between
a vertex in $\mathcal{V}_1$ and that in $\mathcal{V}_2$.
Using the concept of a bipartite graph,
we obtain the following lemma:
\begin{lem}
Consider a connected graph $G=(\mathcal{V}, \mathcal{E})$
satisfying $\#\mathcal{V}>0$ and $\#\mathcal{E}>0$.
Let $\bar{\mathsf X}$ be an undirected incidence matrix of $\bar{G}$.
In this case, $\rank \bar{\mathsf X}< \#\mathcal{V}$ is satisfied if and only if $G$ is 
bipartite.
\end{lem}
\begin{proof}
$\Leftarrow$:
 If the graph is bipartite,
 $\mathcal{V}$ is represented by
 the disjoint union of $\mathcal{V}_1$ and $\mathcal{V}_2$
 as  $\mathcal{V}=\mathcal{V}_1 \sqcup \mathcal{V}_2$.
The $v_i$-component row vector of $\bar{\mathsf X}$ is
 denoted by $\bar{\vct{X}}_{i}$, where
 $v_i \in \mathcal{V}$ $(i=1,2,\cdots, n)$, and
$n=\#\mathcal{V}$.
 Without loss of generality, we can assume  
 $v_i \in \mathcal{V}_1$  $(i=1,2,\cdots, l)$ and
 $v_i\in \mathcal{V}_2$  $(i=l+1, l+2,\cdots, n)$.
Because the graph is bipartite,
$\sum_{i=1}^{l}\bar{\vct{X}}_{i}=\sum_{i=l+1}^{n}\bar{\vct{X}}_{i}$
 is satisfied.\\
 $\Rightarrow$:
 If we assume $\rank \bar{\mathsf X}\ne \#\mathcal{V}$,
\begin{equation}
\sum_{i=1}^{n}c_i \bar{\vct{X}}_{i}=0 \label{eq:19}
\end{equation} is satisfied for
 $c_i \in \mathbb{R}$ and not every $c_i=0$.
 At first we show $c_i\ne 0$ for all $i\in\{1,2,\cdots,n\}$.
We assume $c_{q}=0$ for $q \in \{1,2,\cdots,n\}$.
For arbitrary $r \in \{1,2,\cdots,n\}$,
we consider a path $(v_{f(1)}, e_{1}, v_{f(2)}, e_2, \cdots, e_{p}, v_{f(p+1)})$
 from $v_{q}$ to $v_{r}$, where $p$ is the path length, $f$ is a function
 from $\{1,2,\cdots, p+1\}$ to $\{1,2,\cdots, n\}$ satisfying
 $f(1)=q$ and $f(p+1)=r$, and
the edge $e_{s}$ connects vertices $v_{f(s)}$ and $v_{f(s+1)}$
 $(s=1,2,\cdots, p)$.
 Considering $e_{s}$-column of Eq.~(\ref{eq:19}),
 we have $c_{f(s)}=-c_{f(s+1)}$. Then,
 we obtain $c_i=0$ for all $i$. This leads a contradiction because we assumed not every $c_i=0$.
Therefore, we have $c_i\ne0$ for all $i$.
We assume $c_i>0$ $(i=1,2,\cdots, l)$, and
 $c_i<0$ $(i=l+1,l+2,\cdots, n)$, without loss of generality.
For a given arbitrary $e\in \mathcal{E}$,
we consider the column component about $e$ of Eq.~(\ref{eq:19}).
 Then, we find $e$ is connected between
 a vertex in $\mathcal{V}_1=\{v_1,v_2,\cdots,v_l\}$
 and that in $\mathcal{V}_2=\{v_{l+1},v_{l+2},\cdots,v_{n}\}$.
 This shows $G$ is bipartite.
 This proof is based on Ref.~\onlinecite{VanNuffelen1976}.
  \end{proof}
From this lemma, we can prove the following theorem:
\begin{theo}
 Consider  a simple $m$-regular connected graph $G=(\mathcal{V}, \mathcal{E})$ with $\#\mathcal{E}>0$.
There exists at least one mode with $\omega/\omega_0=\sqrt{2m}$
 for the inductor-capacitor circuit network on $G$
 if and only if $G$ is bipartite.
\end{theo}
\begin{proof}
 Let $\bar{\mathsf X}$ be an undirected incidence matrix of $\bar{G}$.
$G$ is bipartite $\Leftrightarrow$ $\rank{\bar{\mathsf X}}<\#\mathcal{V}$
 $\Leftrightarrow$ $\dim \ker{\bar{\mathsf{X}}^\mathrm{T}}=\#\mathcal{V}-\rank{\bar{\mathsf X}}>0$  $\Leftrightarrow$
  $^\exists \vct{x}$ satisfying $\bar{\mathsf{X}}^\mathrm{T}\vct{x}=0$
 $\Leftrightarrow$
 $^\exists \vct{x}$ satisfying $\mathsf{L}\vct{x}=2m\vct{x}$.
 \end{proof}
 On the other hand, we can formulate the condition
 for presence of eigenvalue $E=2m$ of $\mathsf{L}\sub{L}$ as follows: 
 \begin{theo}
Consider a simple $m$-regular graph $G=(\mathcal{V}, \mathcal{E})$ with $\#\mathcal{V}>0$.
There exists at least one mode with $\omega/\omega_0=\sqrt{2m}$ for the inductor-capacitor circuit network on $L(G)$
 if $m> 2$.
 \end{theo}
\begin{proof}
 Let $\bar{\mathsf X}$ be an undirected incidence matrix of $\bar{G}$.
We have $\dim \ker{\bar{\mathsf X}}=\#\mathcal{E}-\rank{\bar{\mathsf{X}}}\geq \#\mathcal{E}-\#\mathcal{V}=\frac{m-2}{2}\#\mathcal{V}>0$ for $m>2$.
 Then, $^\exists \vct{x}$ satisfying $\bar{\mathsf{X}}\vct{x}=0$.
 Finally,  $^\exists \vct{x}$ satisfying $\mathsf{L}\sub{L}\vct{x}=2m\vct{x}$.
 \end{proof}
From these theorems,
 all spectra of the inductor-capacitor circuit networks on
 a simple $m$-regular ($m>2$) connected bipartite graph and its line graph
 completely coincide.
On the other hand,
we find that $\mathcal{L}\sub{c}$ have a non-degenerated
eigenvector with eigenvalue $E=2m$
for a simple $m$-regular ($m>2$) connected non-bipartite graph $G$
(e.g.  Sec.~\ref{subsub:1}).
For $m=2$, all spectra of the inductor-capacitor circuit networks on a
simple connected $m$-regular graph $G$ and its line graph $L(G)$ coincide
because of $G=L(G)$. In this case,
there are $\omega/\omega_0=\sqrt{2m}$ modes if and only if $\#\mathcal{V}$ is even.

 To analyze an infinite periodic graph with lattice vectors $\{\thrvct{a}_i\}$,
 we take a supercell spanned by $\{N_i \thrvct{a}_i\}$, $N_i>1$.
 We impose a periodic boundary condition (called Born--von Karman boundary condition)
 on the sides of the supercell.
Note that this boundary condition just leads to discretization of wave vectors in the Brillouin zone.
Now, the infinite graph is reduced to a finite graph and
we can use the theorems.
For example, the hexagonal lattice is a simple 3-regular connected bipartite graph.
Then, the spectra of the inductor-capacitor circuit networks on hexagonal and \kagome lattices
include $\omega/\omega_0=\sqrt{6}$, simultaneously.
Note that the spectra of the inductor-capacitor circuit networks on hexagonal and \kagome lattices completely coincide,
but their dispersion relations do not.

 \section{Detailed theoretical model for the metallic hexagonal and \kagome lattices \label{sec:appB}}
In this appendix, we derive detailed theoretical models for fitting the eigenfrequencies of the metallic hexagonal and \kagome lattices.
For the metallic hexagonal and \kagome lattices, the circuit models treated in Sec.~\ref{sec:2} are approximately valid.
To improve the model accuracy, we have to take into account capacitive couplings between disks.
Considering only nearest neighboring couplings,
we modify Eq.~(\ref{eq:4}) as $\mathsf{P}=C^{-1}( \mathsf{I} + \eta \mathsf{A})$.
Here, $C^{-1} \eta \mathsf{A}$ represents
the capacitive coupling between the adjacent disks \cite{Kajiwara2016}.
Then, we obtain
$$ \frac{\omega}{\omega_0}=\sqrt{ \big[m-\alpha(\thrvct{k}_\parallel)\big]\big[1+\eta \alpha(\thrvct{k}_\parallel)\big] }$$
for the metallic hexagonal lattice and
$$ \frac{\omega}{\omega_0}=\sqrt{ \big[m\sub{L}-\alpha\sub{L}(\thrvct{k}_\parallel)\big]\big[1+\eta \alpha\sub{L}(\thrvct{k}_\parallel)\big] }$$
for the metallic \kagome lattice.
Generally, the coupling constant $\eta$ depends on $\omega$ as $\eta=\eta_0 \exp[\ii (\omega/c) b]$, where $b$ is the distance between the nearest disks \cite{Yeung2011}.
The imaginary part of $\eta$ represents the resistive component.
If we only focus on the real part of the eigenfrequencies, we may ignore the resistive term (note that we have already ignored the imaginary part of $L$ and $C$).
Then, we assume $\eta=\eta_0 \cos[(\omega/c) b]$.
Using the real part of the eigenvalues calculated by simulation, we numerically minimize the error
$$\mathrm{err}\sur{hex}(\omega_0, \eta_0)=\sum_{(\omega,\thrvct{k}_\parallel) \in \text{\{data points\}} }g\big(\omega_0, \eta_0;\omega, m, \alpha(\thrvct{k}_\parallel)\big)^2$$ 
for the hexagonal lattice and 
$$\mathrm{err}\sur{kag}(\omega_0, \eta_0)=\sum_{(\omega,\thrvct{k}_\parallel) \in \text{\{data points\}} }g\big(\omega_0, \eta_0;\omega, m\sub{L}, \alpha\sub{L}(\thrvct{k}_\parallel)\big)^2$$ 
for the \kagome lattice, where we define 
$$g(\omega_0, \eta_0;\omega, m, \alpha)=\omega- \omega_0\sqrt{\left[m-\alpha\right]\left[1+\alpha \eta_0 \cos\left(\frac{\omega b}{c}\right)\right]}.$$
The obtained fitting parameters are as follows: $\omega_0=\omega_0\sur{hex}=2\pi \times 0.107\,\U{THz}$, $\eta_0=\eta\sur{hex}_0=0.0916$ for the hexagonal lattice, and $\omega_0=\omega_0\sur{kag}=2\pi\times 0.101\,\U{THz}$ and $\eta_0=\eta\sur{kag}_0=0.142$ for the \kagome lattice.
Because the magnetic coupling and higher order effects (beyond the nearest capacitive coupling) are included in these parameters, the parameters for the hexagonal and \kagome lattice can be different.
The dispersion curves with these fitting parameters are shown in Fig.~\ref{fig:hex_kagome_bands}.
These curves agree with the simulated data despite the simplicity of the model.

\nocite{*}

%

\end{document}